\documentclass[10pt,journal,letterpaper]{IEEEtran}
 
\usepackage{psfrag}
\usepackage{epsfig}
\usepackage{graphicx}
\usepackage{multicol}
\usepackage{cite} 
\usepackage{amsthm}
\usepackage{subfigure}
\usepackage{nicefrac}

\RequirePackage{amsfonts}
\RequirePackage{amssymb}
\RequirePackage{amsopn}
\RequirePackage{amsmath}
\RequirePackage{pifont}
\RequirePackage{mathrsfs}
\RequirePackage{pstricks}

\def\atauxout{\csname @auxout\endcsname}%
\def\labelii#1{\immediate\write\atauxout{%
    \noexpand\newlabel{#1}{{\theenumii}{\thepage}}}}

\newcommand{\vect}[1]{\mathbf{#1}} 
 
\newcommand{\Integers}{\mathbb Z}   
\newcommand{\Complex}{\mathbb C}    

\newcommand{\bra}{\langle} 
\newcommand{\ket}{\rangle} 

\newcommand{\E}[2][]{\textnormal{\textsf{E}}_{#1}\!\left[#2\right]} 

\newcommand{\markov}{\textnormal{\mbox{$\multimap\hspace{-0.73ex}-\hspace{-2ex}-$}}}

\renewcommand{\d}{\,\textnormal{d}}

\newbox\measurebox %
\newlength{\firstmini} %
\newcommand{\mytextandeps}[2]{%
  \setbox\measurebox\hbox{\epsfig{file=#2}}
  \setlength{\firstmini}{\linewidth}%
  \addtolength{\firstmini}{-\wd\measurebox}%
  \addtolength{\firstmini}{-1em}%
  \begin{minipage}[t]{\firstmini}
    #1
  \end{minipage} \hfill %
  \setbox\measurebox\vbox{\unhbox\measurebox} %
  \setlength{\firstmini}{\ht\measurebox} %
  \addtolength{\firstmini}{\dp\measurebox} %
  \ht\measurebox=0pt \dp\measurebox=\firstmini %
  \box\measurebox
}%
  
\newgray{mygray}{0.9}
\makeatletter
  {\egroup     
  \noindent
  \centerline{
  \psshadowbox[shadowsize=0.3em,framesep=1.5em, fillstyle=solid, fillcolor=mygray]
              {\box\@tempboxa}
              }
  \par
  \vskip 0pt plus1ex 
  }%
\makeatother

\makeatletter
  {\egroup     
  \noindent
  \centerline{
  \psshadowbox[shadowsize=0.3em,framesep=1.5em, fillstyle=solid, fillcolor=white]
              {\box\@tempboxa}
              }
  \par
  \vskip 0pt plus1ex 
  }%

\newtheorem{proposition}{Proposition}

\newtheorem{scheme}{Scheme}

\newcounter{sc}

\newcommand{\e}{\mathop{}\!e}

\newcommand{\ave}{\mathcal{E}}

\title{Toward Photon-Efficient Key Distribution over Optical Channels}

\author{Yuval Kochman,~\IEEEmembership{Member,~IEEE}, Ligong Wang,~\IEEEmembership{Member,~IEEE}, and Gregory
  W. Wornell~\IEEEmembership{Fellow,~IEEE}}

\begin{document}

\maketitle
\renewcommand{\thefootnote}{} 
\footnotetext[1]{This work was supported in part by the DARPA InPho
  program
 under Contract No. HR0011-10-C-0159, and by AFOSR under Grant
  No.~FA9550-11-1-0183. Y. Kochman was also supported in part by Israel Science Foundation under Grant No.~956/12.

Y. Kochman is with the School of Computer Science and Engineering, Hebrew
University of Jerusalem, Israel; he was with the Research Laboratory of Electronics, 
    Massachusetts Institute of Technology, Cambridge, MA, 
    USA (e-mail: yuvalko@cs.huji.ac.il).

L. Wang and G.W. Wornell are with the Research Laboratory of Electronics, 
    Massachusetts Institute of Technology, Cambridge, MA, 
    USA (e-mail: wlg@mit.edu; gww@mit.edu).}
\renewcommand{\thefootnote}{\arabic{footnote}}

\begin{abstract}
This work considers the distribution of a secret key over an optical (bosonic) channel in the regime of high photon efficiency, i.e., when the number of secret key bits generated per detected photon is high. While in principle the photon efficiency is unbounded, there is an inherent tradeoff between this efficiency and the key generation rate (with respect to the channel bandwidth). 
We derive asymptotic expressions for the optimal generation rates in the photon-efficient limit, and propose schemes that approach these limits up to certain approximations. The schemes are practical, in the sense that they use coherent or temporally-entangled optical states and direct photodetection, all of which are reasonably easy to realize in practice, in conjunction with off-the-shelf classical codes.
\end{abstract}

\begin{IEEEkeywords}
Information-theoretic security, key distribution, optical communication, wiretap channel.

\end{IEEEkeywords}

\section{Introduction}

\IEEEPARstart{I}{nformation-theoretic} key distribution \cite{ahlswedecsiszar93,maurer93} involves the generation of a sequence between the participating terminals, such that the mutual information between this sequence and any data obtained by other terminals is close to zero in an appropriate sense. Unlike secure communication through the wiretap channel \cite{wyner75}, the sequence need not be known \emph{a priori} to any of the terminals. Like the latter, however, the information-theoretic approach to key distribution hinges on knowledge of the channel through which an adversarial terminal listens to the communication, as opposed to computational approaches where the assumption is the inability of the adversary to perform certain computations in reasonable time. The computational hardness assumption may no longer be valid when future technology, e.g., quantum computers, becomes available, causing the computational approaches to fail. But the information-theoretic approach also has its drawback: the information obtained by the legitimate terminals cannot prove or disprove the channel assumption on which the key-distribution protocol is based, inhibiting security in a realistic setting.

The situation is much different when a quantum channel is employed \cite{bennettbrassard84,ekert91}. Loosely speaking, the ``no-cloning'' theorem \cite{wootterszurek82} guarantees that information ``stolen'' by an eavesdropper will not reach the legitimate terminal, thus the situation where the adversary is stronger than initially assumed can be detected. In fact, even eavesdroppers that can actively transmit into the quantum channel can be detected, at the cost of key-rate loss, using measurements based on local randomness. We shall come back to these issues in the discussion at the end of the paper. For the main part of the paper, we rely on the existence of good detection methods to assume that the eavesdropper is passive, and that the complete statistical characterization of the eavesdropper's channel is known to the legitimate terminals.

Two-terminal quantum key distribution (QKD) protocols can be roughly divided into two classes. In ``prepare and measure'' protocols, one legitimate terminal (Alice) prepares  quantum states that are sent via a quantum channel to the other terminal (Bob) and to the eavesdropper (Eve). In contrast, in entanglement-based protocols, a quantum source emits entangled states, which are observed by all terminals via quantum channels. These two classes are parallel to the ``C'' (channel) and ``S'' (source) models of \cite{ahlswedecsiszar93}; in this work we shall use the C/S notation. In either approach, the quantum stage is followed by the use of a classical communication channel. This channel is assumed to be public, i.e., all information sent is received by Eve; however, it is assumed that Eve cannot transmit into this public channel. The performance of a QKD scheme is measured in terms of the size of the secret key normalized by the quantum-channel resources used. The classical channel is thus ``free'', although its use is limited by the assumption that Eve has full access to this channel. 

A quantum channel most often encountered in practice is the optical channel, which is modeled in quantum mechanics as a bosonic channel. When used for communicating classical data at low average input power, it is asymptotically optimal to use a direct-detection receiver, which ignores the phase of the optical signal. This results in an equivalent classical channel where the output has a Poisson distribution whose mean is proportional to the channel's input \cite{shapiro09}. Some of the first important works on this channel model are in \cite{bardavid69,kabanov78,davis80}. 
The low-input-power regime can be thought of as a ``photon-efficient regime''. This is because, in the limit of low average photon number per channel use, the communication rate per photon is unbounded.

In this work we consider QKD over the bosonic channel in the photon-efficient regime. We consider both C and S models, and show that in both, as happens in communication, the photon efficiency is unbounded and direct-detection receivers are asymptotically optimal. We further consider specific QKD protocols. We discuss the difficulty of finding code constructions that allow us to approach the theoretical performance limits, since in the photon-efficient regime they have to operate over highly-skewed sequences. We present protocols that overcome this difficulty: in the C model we use pulse-position modulation (PPM), while in the S model we parse the sequence of detections into frames. In both cases, coding over frames is an easier task than coding directly over the detection sequence. 

The rest of the paper is organized as follows. We introduce our notation in Section~\ref{sec:notation}. In Section~\ref{sec:setting} we formally describe the problem setting. Then in Section~\ref{sec:communication} we discuss, as a point of reference, photon-efficient communication. Sections~\ref{sec:channel_case} and \ref{sec:source_case} include our main results for key distribution, regarding the C and S models, respectively. We conclude this paper in Section~\ref{sec:discussion} by discussing the gap between our results and fully quantum security proofs.

\section{Notation}\label{sec:notation}

We use a font like $\mathbb{A}$ to denote a Hilbert space. Throughout this paper we shall focus on bosonic Hilbert spaces. We adopt Dirac's notation to use $|\psi\ket$ to denote a unit vector in a Hilbert space, which can describe a pure quantum state, and use $\bra \psi|$ to denote the conjugate of $|\psi\ket$. We follow most of the physics literature to slightly abuse our notation: we shall not make typographical distinctions between number states and coherent states. Hence $|n\ket$, $n\in\Integers_0^+$, (usually) denotes the number state that contains $n$ photons; while $|\alpha\ket$, $\alpha\in\Complex$, (almost everywhere) denotes a coherent state, whose exact characterization is given later. This abuse of notation will not cause confusion within the scope of this paper. We use a Greek letter like $\rho$ to denote a density operator (i.e., a trace-one semidefinite operator) on a Hilbert space, which can describe a pure or mixed quantum state. Note that the density-operator description of a pure state $|\psi\ket$ is $|\psi\ket\bra\psi|$. When considering a system such as a beamsplitter, we reserve the letters $|\psi\ket$ and $\rho$ for input states, and $|\phi\ket$ and $\sigma$ for output states. Sometimes, to be explicit, we add a superscript to a state to indicate its Hilbert space so it looks like $|\psi\ket^\mathbb{A}$ or $\sigma^\mathbb{B}$. We use the notation $\hat{a}$ to denote the annihilation operator on $\mathbb{A}$ (so $\hat{a}^\dag$ is the creation operator on $\mathbb{A}$); similarly, $\hat{b}$ denotes the annihilation operator on $\mathbb{B}$, etc.

For a quantum state $\sigma^\mathbb{AB}$ on the Hilbert spaces $\mathbb{A}$ and $\mathbb{B}$, we use $H(\sigma^\mathbb{A})$, $H(\sigma^\mathbb{A}|\sigma^\mathbb{B})$, and $I(\sigma^\mathbb{A};\sigma^\mathbb{B})$ to denote the corresponding entropy, conditional entropy, and mutual information, respectively. These quantities are defined as follows (see \cite{nielsenchuang00} for more details):
\begin{IEEEeqnarray}{rCl}
	H(\sigma^\mathbb{A}) & \triangleq & - \textnormal{tr} \left\{ \sigma^\mathbb{A} \log \sigma^\mathbb{A} \right\}\\
	H(\sigma^\mathbb{A}| \sigma^\mathbb{B}) & \triangleq & H(\sigma^{\mathbb{AB}} ) - H(\sigma^\mathbb{B})\\
	I(\sigma^\mathbb{A};\sigma^\mathbb{B}) & \triangleq & H(\sigma^\mathbb{A})+ H(\sigma^\mathbb{B})- H(\sigma^{\mathbb{AB}}).
\end{IEEEeqnarray}
For classical or mixed classical-quantum states, we simply replace the density operator by the classical random variable for the classical part in these expressions, so they look like, e.g., $H(X)$, $H(X|\sigma^\mathbb{B})$, and $I(\sigma^\mathbb{A}; Y)$. Sometimes, to be more precise, we also write the mutual information as $I(\mathbb{A};\mathbb{B})\vert_\sigma$, indicating that it is the mutual information between space $\mathbb{A}$ and $\mathbb{B}$ evaluated for the joint state $\sigma$. 

Throughout this paper, we use natural logarithms, and measure information in nats, though sometimes we do talk about ``bits'' and ``binary representation''.

We use the usual notation $O(\cdot)$ and $o(\cdot)$ to describe behaviors of functions of $\ave$ in the limit where $\ave$ approaches zero with other variables, if any, held fixed. Specifically, given a reference function $f(\cdot)$ (which might be the constant $1$), a function denoted as $O(f(\ave))$ satisfies
\begin{equation}
	\varlimsup_{\ave\downarrow 0} \left|\frac{O(f(\ave))}{f(\ave)} \right|< \infty,
\end{equation}
while a function denoted as $o(f(\ave))$ satisfies
\begin{equation}
	\lim_{\ave\downarrow 0} \frac{o(f(\ave))}{f(\ave)} = 0.
\end{equation}


\section{Problem Setting}\label{sec:setting}

In this section we 
describe our setups for optical communication and key distribution. To do so, we first recall some basic results in quantum optics.

\subsection{Beamsplitting and Direct Detection}

We briefly describe how \emph{number (Fock) states} and \emph{coherent states} evolve when passed through a beamsplitter, and what outcomes they induce when fed into a direct-detection receiver, i.e., a photon counter. We refer to \cite{mandelwolf95} for more details. For some background in quantum physics and in quantum information theory, we refer to \cite{nielsenchuang00}.

Let $\mathbb{A}$ and $\mathbb{V}$ be the two input spaces to a single-mode beamsplitter, and $\mathbb{B}$ and $\mathbb{E}$ be the two output spaces. Let the beamsplitter's transmissivity from $\mathbb{A}$ to $\mathbb{B}$ be $\eta\in[0,1]$. Then this beamsplitter is characterized in the Heisenberg picture by
\begin{subequations}\label{eq:beamsplitter}
\begin{IEEEeqnarray}{rCl}
  \hat{b} & = & \sqrt{\eta}\,\hat{a} + \sqrt{1-\eta}\,\hat{v} \label{eq:beamsplitter1}\\
  \hat{e} & = & \sqrt{1-\eta}\,\hat{a} - \sqrt{\eta}\,\hat{v}.
\end{IEEEeqnarray}
\end{subequations}
Throughout this paper we shall only consider situations where the second input space $\mathbb{V}$ (the ``noise mode'') is in its vacuum state $|0\ket$. 

Ideal direct detection (i.e., photon counting) measures an optical state in the number-state basis. For direct detection on $\mathbb{A}$, the observable is the Hermitian operator $\hat{a}^\dagger \hat{a}$. On state $\rho$, a photon counter gives outcome $n\in\Integers_0^+$ with probability $\bra n|\rho|n\ket$. 

Obviously, when a number state $|n\ket$, $n\in\Integers_0^+$, is fed into an ideal photon counter, the outcome is $n$ with probability one. But passing $|n\ket$ through a beamsplitter is more complicated: if space $\mathbb{A}$ in \eqref{eq:beamsplitter} is in state $|n\ket$, then the  output state is an entangled state on $\mathbb{B}$ and $\mathbb{E}$:
\begin{equation}
	|\phi\ket^{\mathbb{BE}} = \sum_{i=0}^n \sqrt{\binom{n}{i}} \,\eta^{i/2}(1-\eta)^{(n-i)/2} |i\ket^\mathbb{B} |n-i\ket^\mathbb{E}.
\end{equation}
This implies that performing direct detection on the output of this beamsplitter will yield a binomial distribution on the outcome: the probability of detecting $m$ photons on space $\mathbb{B}$ is
\begin{equation} \label{eq:number_direct}
	\bra m | \sigma^\mathbb{B} |m \ket = \binom{n}{m}\eta^m (1-\eta)^{n-m}
\end{equation}
for $0\leq m \leq n$, and is zero otherwise. It also implies that, if direct detection is performed both on $\mathbb{B}$ and on $\mathbb{E}$, then with probability one the sum of the two outcomes is equal to $n$.

A coherent state $|\alpha\ket$, $\alpha\in\Complex$, can be written in the number-state basis as
\begin{equation}\label{eq:coherent_state}
	|\alpha\ket = \e^{-|\alpha|^2/2} \sum_{n=0}^\infty \frac{\alpha^n}{\sqrt{n!}} |n\ket.
\end{equation}
Thus, when fed into a photon counter,  the probability of $n$ photons being observed in $|\alpha\ket$ is 
\begin{equation}
	\bra n|\alpha\ket\bra \alpha |n\ket = |\bra n|\alpha\ket|^2 = \e^{-|\alpha|^2} \frac{|\alpha|^{2n}}{n!}.
\end{equation} 
Namely, the number of photons in $|\alpha\ket$ has a Poisson distribution of mean $|\alpha|^2$.

Coherent states have the nice property that, when passed through a beamsplitter, the outcomes remain in coherent states. If $|\alpha\ket$ is fed into the beamsplitter~\eqref{eq:beamsplitter}, the output state is
\begin{equation}
	|\phi\ket^{\mathbb{BE}} = |\sqrt{\eta}\,\alpha\ket^{\mathbb{B}}\otimes |\sqrt{1-\eta}\,\alpha\ket^{\mathbb{E}}.
\end{equation}
Therefore, if direct detection is performed both on $\mathbb{B}$ and on $\mathbb{E}$, the outcomes will be two \emph{independent} Poisson random variables of means $\eta|\alpha|^2$ and $(1-\eta)|\alpha|^2$, respectively.

\subsection{Optical Communication} \label{sec:comm_background}

A single-mode pure-loss optical (i.e., bosonic) channel can be described using the beamsplitter~\eqref{eq:beamsplitter1}, where we ignore the output space $\mathbb{E}$ and assume the noise space $\mathbb{V}$ to be in its vacuum state. In this formula, $\mathbb{A}$ is the input space controlled by the transmitter which, in consistency with the key-distribution part, we call Alice;  $\mathbb{B}$ is the output space obtained by the receiver, Bob; and $\eta$ is the transmissivity of the channel. Equivalently, the channel may be described in the Schr\"odinger picture as a completely-positive trace-preserving (CPTP) map from the input state $\rho^\mathbb{A}$ to the output state $\sigma^\mathbb{B}$:
\begin{equation}\label{eq:channel_map}
\sigma^{\mathbb{B}} = \mathcal{C}(\rho^{\mathbb{A}}).
\end{equation}
The explicit characterization of $\mathcal{C}$ is complicated and omitted.

We denote the blocklength of a channel code by $k$. Alice has a message of $kR$ nats\footnote{We ignore the fact that the number of values that the message can take is not an integer.} to convey to Bob. In order to do this, she prepares a state $\rho^k$ over  $\mathbb{A}^k$, subject to an average-photon-number constraint $\ave$ per channel use:
\begin{equation}\label{eq:ave_block}
  \textnormal{tr}\left\{ \left(\sum_{i=1}^k \hat{a}_i^\dag \hat{a}_i \right) \rho^k \right\} \le k \ave
\end{equation}
where $\hat{a}_i$ is the annihilation operator on the input space of the $i$th channel use. The channel is assumed to be memoryless, so the output is given by
\begin{equation}\label{general_channel}
\sigma^k = \mathcal{C}^{\otimes k} (\rho^k).
\end{equation} 
Bob may perform any positive-operator valued measure (POVM) on $\sigma^k$ to reconstruct the message. As usual, the capacity of the channel is defined as the supremum of rates for which there exist sequences of schemes with increasing blocklengths and with the error probabilities approaching zero.

We define the \emph{photon efficiency} of transmission as the rate normalized by the expected number of photons that Bob receives per channel use:\footnote{We adopt this definition rather than normalizing by transmitted photons, because this allows us to derive expressions which are less influenced by the transmissivity of the channel.}
\begin{equation} \label{eq:efficiency}
r (\eta,\ave) \triangleq \frac{R(\eta,\ave)}{\eta \ave}.
\end{equation}
This quantity is upper-bounded by the channel's capacity divided by $\eta \ave$.

\subsection{Key Distribution Using an Optical Channel (Model C)} \label{sec:key_channel_setup}

We next consider the problem where Alice and Bob use the channel of \eqref{eq:beamsplitter} to generate a secret key between them. The channel from Alice to Bob is still characterized by \eqref{eq:beamsplitter1} or by the CPTP~\eqref{eq:channel_map}, but we now assume that an eavesdropper, Eve, obtains the Hilbert space $\mathbb{E}$. Note that this is a worst-case assumption in the sense that Eve obtains the whole \emph{ancilla} system of the channel. Also note that we assume Eve to be passive, so she cannot interfere with the communication; she can only try to distill useful information from her observations. This setting can be seen as a special case of the quantum version of ``Model C'' discussed in \cite{ahlswedecsiszar93}.

The aim of Alice and Bob is to use this channel, together with a two-way, public, but authentic classical channel, to generate a secret key. Let $k$ denote the total number of uses of the optical channel. We impose the same average-photon-number constraint \eqref{eq:ave_block} on Alice's inputs. We assume the public channel is free so we can use it to transmit as many bits as needed, though all these bits will be known to Eve. By the end of a key-distribution protocol,  Alice should be able to compute a bit string $S_A$ and Bob should be able to compute $S_B$ such that
\begin{itemize}
  \item The probability that $S_A=S_B$ tends to one as $k$ tends to infinity;
  \item The key $S_A$ (or $S_B$) is almost uniformly distributed and independent of Eve's observations, in the sense that $$\frac{H(S_A|  \rho_{\textnormal{Eve}})}{\log|\mathcal{S}|}$$ tends to one as $k$ tends to infinity, where $\rho_\textnormal{Eve}$ summarizes all of Eve's observations, and where $\mathcal{S}$ denotes the alphabet for $S_A$ and $S_B$. 
\end{itemize}

We define the \emph{secret-key rate} of a scheme to be
\begin{equation}
	R(\eta,\ave) \triangleq \frac{\log |\mathcal{S}|}{k}
\end{equation}
nats per use of the optical channel. The parameter $\ave$ is the average photon number in~\eqref{eq:ave_block}.

A typical (and rather general) protocol to accomplish this task consists of the following steps:

\textbf{Step 1:} Alice generates random variables $X_1,X_2,\ldots$ which are known to neither Bob nor Eve. She then prepares an optical state $\rho^k$ on $\mathbb{A}^k$ based on $\vect{X}$ and sends the state into the channel, spread over $k$ channel uses.

\textbf{Step 2:} Bob makes measurements on his output state to obtain a sequence $Y_1,Y_2,\ldots$.\footnote{We do not consider feedback from Bob to Alice during the first two steps. As in channel coding, feedback cannot increase the maximum key rate.}

\textbf{Step 3:} (Information Reconciliation) Alice and Bob exchange messages $M_1,M_2,\ldots$
using the public channel. Then Alice computes her raw key $K_A$ as a
function of $(\vect{X},\vect{M})$, and Bob computes his raw key $K_B$
as a function of $(\vect{Y},\vect{M})$. They try to ensure that $K_A=K_B$ with high probability, but Eve might have partial information about the raw key.

\textbf{Step 4:} (Privacy Amplification) Alice and Bob randomly pick one from a set of universal hashing functions. They apply the chosen function to their raw keys $K_A$ and $K_B$ to obtain the secret keys $S_A$ and $S_B$, respectively.

Privacy amplification has been extensively studied in literature. Denote the quantum state that Eve obtained in Step~1 from the optical channel by $\sigma^{\mathbb{E}^k}$. It is shown in \cite{rennerkonig05} that, provided $K_A=K_B$ with probability close to one, the privacy amplification step (i.e., Step~4) can be accomplished successfully with high probability, and the length of the secret key in nats, i.e., $\log|\mathcal{S}|$, can be made arbitrarily close to\footnote{To be precise, to achieve \eqref{eq:rennerkonig}, Alice and Bob should repeat Steps~1 to~3
  many times, and then do Step~4 on all the raw keys together.}
\begin{equation}\label{eq:rennerkonig}
  H(K_A|\vect{M},\sigma^{\mathbb{E}^k}).
\end{equation}
Hence, in this paper, we shall not discuss how to accomplish Step~4. As we shall see, in some cases Step~4 can be omitted. If not, then we shall concentrate on Steps~1 to~3, try to maximize~\eqref{eq:rennerkonig}, and compute the secret-key rate as
\begin{equation}\label{eq:key_rate}
	R(\eta,\ave) = \frac{H(K_A|\vect{M},\sigma^{\mathbb{E}^k})}{k}.
\end{equation}

As mentioned previously, in Step~1, we impose the same average-photon-number constraint on Alice \eqref{eq:ave_block} as in the communications case. Consequently, we define the photon efficiency (of key distribution) $r(\eta,\ave)$ in the same way as in communications, namely, as in \eqref{eq:efficiency}, except that now $R(\eta,\ave)$ is the secret-key rate.

\subsection{Key Distribution Using a Photon Source (Model S)}\label{sec:key_source_setup}

In some key-distribution protocols, Alice and Bob make use of a random source, rather than Alice preparing states, to generate a secret key, as in the ``Model S'' discussed in~\cite{ahlswedecsiszar93}. In optical applications one can, for instance, generate a uniform stream of random, temporally-entangled photon pairs, which are very useful for key distribution; see, e.g., \cite{zhongwong12}.

An accurate model for such temporally-entangled photon sources divides the timeline into very fine temporal modes, where each temporal mode is in a pure, entangled state on its two output Hilbert spaces, with the number of photon pairs having a geometric (Bose-Einstein) distribution of a very small mean. Such a model, however, would be intractable for precise key-rate analyses. We hence choose a simplified model as follows. Let the timeline be divided into slots, where each slot can be thought of as one ``use'' of the source. Each slot contains many, e.g., a thousand, temporal modes. This results in the number of photon pairs in each slot having a Poisson distribution, whose mean $\ave$ equals the total number of temporal modes times the mean photon number in each mode. We ignore the fine structures inside each slot and describe it with only two Hilbert spaces, $\mathbb{C}$ and $\mathbb{D}$. We also ignore the entanglement between the two spaces and simplify the optical state to a mixed one with classical correlation only. The optical state emitted by the source in every source use is thus given by 
\begin{equation}\label{eq:source_model}
	\rho^{\mathbb{CD}} = \sum_{i=0}^\infty \frac{\ave^i \e^{-\ave}}{i!}  |i\ket \bra i|^\mathbb{C} \otimes |i\ket\bra i|^\mathbb{D}.
\end{equation}

To justify the simplification we make, note the following.
\begin{itemize}
\item{Discretization:} In our schemes, Alice and Bob will never measure the arrival time of a photon with higher accuracy than the duration of one slot. In this case, it is easy to show that Eve cannot have any advantage by making finer measurements.
\item{Classical correlation:} When Alice and Bob only make direct detection with the given time accuracy and Eve is listening passively via a beamsplitter, entanglement does not play any role. Note that this would not be the case if we were interested in a secrecy proof against a general (possibly active) Eve; see Section~\ref{sec:discussion}.
\end{itemize}

We assume that the source is collocated with Alice, who keeps $\mathbb{C}$; while the photons in $\mathbb{D}$ are sent to Bob through a lossy optical channel. To account for coupling losses, we can assume that Alice also only has access to a lossy version of $\mathbb{C}$. Specifically, $\rho^\mathbb{C}$ is passed through a beamsplitter, like the one in \eqref{eq:beamsplitter}, of transmissivity $\eta_A$ before it reaches Alice:
\begin{subequations}\label{eq:beamsplitterA}
\begin{IEEEeqnarray}{rCl}
  \hat{a} & = & \sqrt{\eta_A}\,\hat{c} + \sqrt{1-\eta_A}\,\hat{v}\\
  \hat{f} & = & \sqrt{1-\eta_A}\,\hat{c} - \sqrt{\eta_A}\,\hat{v}.
\end{IEEEeqnarray}
\end{subequations}
But, except for Section~\ref{sec:etaA}, we shall ignore coupling losses and take $\eta_A=1$. 
Similarly $\rho^\mathbb{D}$ is passed through a beamsplitter of transmissivity $\eta_B$ before it reaches Bob:
\begin{subequations}\label{eq:beamsplitterB}
\begin{IEEEeqnarray}{rCl}
  \hat{b} & = & \sqrt{\eta_B}\,\hat{d} + \sqrt{1-\eta_B}\,\hat{u} \\
  \hat{e} & = & \sqrt{1-\eta_B}\,\hat{d} - \sqrt{\eta_B}\,\hat{u}.
\end{IEEEeqnarray}
\end{subequations}
We assume $\eta_B<1$ throughout. 
Both noise modes $\mathbb{V}$ and $\mathbb{U}$ are assumed to be in their vacuum states. Note that the two beamsplitters behave independently of each other.

Since the source is collocated with Alice, we know the photons that are lost from $\mathbb{C}$ to $\mathbb{A}$ (in case $\eta_A<1$) should \emph{not} reach Eve; Eve only has access to the Hilbert space $\mathbb{E}$.

It is useful to describe the output states when $\rho^\mathbb{D}$ passes through the beamsplitter on Bob's side. We first write down $\rho^\mathbb{D}$ by taking partial trace of \eqref{eq:source_model}:
\begin{equation}
	\rho^\mathbb{D} = \sum_{i=0}^\infty \frac{\ave^i \e^{-\ave}}{i !} |i\ket\bra i|,
\end{equation}
which can be equivalently written as
\begin{equation}\label{eq:rhoD_coherent}
	\rho^\mathbb{D} = \frac{1}{2\pi} \int_0^{2\pi} \d \theta |\alpha(\theta)\ket\bra\alpha(\theta)|
\end{equation}
where
\begin{equation}
	\alpha(\theta) = \sqrt{\ave} \e^{i\theta}.
\end{equation}
From \eqref{eq:rhoD_coherent} and \eqref{eq:beamsplitterB} it is straightforward to obtain the output optical state on $\mathbb{BE}$:
\begin{IEEEeqnarray}{rCl}
	\sigma^{\mathbb{BE}}	
& = & \frac{1}{2\pi} \int_0^{2\pi} \d \theta |\sqrt{\eta}\, \alpha(\theta)\ket \bra \sqrt{\eta}\, \alpha(\theta) |^\mathbb{B} \nonumber\\
	& & ~~~~~~~~\otimes |\sqrt{1-\eta}\, \alpha(\theta)\ket \bra \sqrt{1-\eta}\, \alpha(\theta)|^\mathbb{E}.\label{eq:sigmaBE}
\end{IEEEeqnarray}
By taking partial traces of \eqref{eq:sigmaBE} we obtain Bob's and Eve's states:
\begin{IEEEeqnarray}{rCl}
	\sigma^\mathbb{B} & = &  \frac{1}{2\pi} \int_0^{2\pi} \d \theta |\sqrt{\eta}\, \alpha(\theta)\ket \bra \sqrt{\eta}\, \alpha(\theta) | \\
	& = & \sum_{i=0}^\infty \frac{(\eta  \ave)^i \e^{-\eta \ave}}{i!}  |i\ket\bra i|\\
	\sigma^\mathbb{E} & = &  \frac{1}{2\pi} \int_0^{2\pi} \d \theta |\sqrt{1-\eta}  \alpha(\theta)\ket \bra \sqrt{1-\eta}  \alpha(\theta) | \\
	& = & \sum_{i=0}^\infty \frac{((1-\eta) \ave)^i \e^{-(1-\eta)\ave}}{i!}  |i\ket\bra i|. \label{eq:sigmaE}
\end{IEEEeqnarray}

The joint state $\sigma^{\mathbb{BE}}$ is not a tensor state, i.e., $\mathbb{B}$ and $\mathbb{E}$ are not independent. However, if direct detection---namely, projective measurement in the number-state basis---is performed on $\mathbb{B}$ (or on $\mathbb{E}$), the post-measurement state on $\mathbb{E}$ (or on $\mathbb{B}$) is independent of the measurement outcome; in particular, the photon numbers in $\mathbb{B}$ and in $\mathbb{E}$ are independent. Indeed, conditional on the measurement outcome on $\mathbb{B}$ being $i$, the post-measurement state on $\mathbb{E}$ is
\begin{IEEEeqnarray}{rCl}
 	\lefteqn{\frac{\textnormal{tr}_\mathbb{B} \left\{ |i\ket\bra i|^\mathbb{B} \sigma^{\mathbb{BE}}\right\}}{\textnormal{tr}\left\{|i\ket\bra i|^\mathbb{B} \sigma^\mathbb{B} \right\}}}\nonumber\\
	 &= & \frac{\displaystyle \frac{1}{2\pi} \int_0^{2\pi} \d \theta |\bra i | \sqrt{\eta}\, \alpha(\theta)\ket|^2 |\sqrt{1-\eta}\, \alpha(\theta)\ket \bra \sqrt{1-\eta}\, \alpha(\theta)|} {\displaystyle \frac{1}{2\pi} \int_0^{2\pi} \d \theta |\bra i | \sqrt{\eta}\, \alpha(\theta)\ket|^2}\nonumber\\ {} \\
	& = & \frac{1}{2\pi} \int_0^{2\pi} \d \theta |\sqrt{1-\eta}\, \alpha(\theta)\ket \bra \sqrt{1-\eta}\, \alpha(\theta) | \label{eq:indep_theta}\\
	& = & \sigma^\mathbb{E}
\end{IEEEeqnarray}
where \eqref{eq:indep_theta} follows because
\begin{equation}
	|\bra i | \sqrt{\eta}\, \alpha(\theta)\ket|^2 = \frac{(\eta \ave)^i \e^{-\eta \ave}}{i!}
\end{equation}
does not depend on $\theta$. 

We now describe a scheme (which is again rather general) for Alice and Bob to use this source $k$ times to generate a secret key. In this scheme, Steps~3 and~4 are exactly the same as in Section~\ref{sec:key_channel_setup}, but Steps~1 and~2 are now replaced by:

\textbf{Step 1':} Alice makes measurements on her state $\sigma^{\mathbb{A}^k}$ to obtain the sequence $X_1,X_2,\ldots$.

\textbf{Step 2':} Bob makes measurements on his state $\sigma^{\mathbb{B}^k}$ to obtain the sequence $Y_1,Y_2,\ldots$.

As in Section~\ref{sec:key_channel_setup}, we shall concentrate on Steps~1', 2', and~3. The secret-key rate, denoted by $R(\eta_A,\eta_B,\ave)$, is again given by the right-hand side of \eqref{eq:key_rate}, with unit ``nats per source use''. But the photon efficiency in this setting is defined as
\begin{equation}\label{eq:PE_etaAB}
	r(\eta_A,\eta_B,\ave) \triangleq \frac{R(\eta_A,\eta_B,\ave)}{\eta_A \eta_B \ave}.
\end{equation}
We choose this definition because $\eta_A \eta_B \ave$ is the expected number of photon pairs in each source use that reach both Alice and Bob,\footnote{We interpret this quantity in a semi-classical way: each photon pair reaches Alice with probability $\eta_A$, and reaches Bob with probability $\eta_B$ independently of whether it reaches Alice or not, hence the fraction of photon pairs that reach both Alice and Bob is $\eta_A \eta_B$. We do not know if there exists a physical observable, i.e., a Hermitian operator that corresponds to this value.} and because these photon pairs are those that contain correlated information that can be used to generate the secret key. When $\eta_A=1$, we omit the subscript in $\eta_B$, and denote the secret-key rate and photon efficiency simply by $R(\eta,\ave)$ and $r(\eta,\ave)$, respectively. Obviously, they are again related by \eqref{eq:efficiency}.


\section{Background: Photon-Efficient Communication using Pulse-Position Modulation}\label{sec:communication}

Before we address key distribution, we give some results regarding communications over the bosonic channel described in Section~\ref{sec:comm_background}. 
These results serve as a point of reference, and the derivation provides tools later used in key distribution. See also \cite{kochmanwornell12,erkmen12,wangwornell14}.

The capacity of a quantum channel is characterized by the formula found by Holevo \cite{holevo98} and by Schumacher and Westmoreland \cite{schumacherwestmoreland97}. For the pure-loss bosonic channel \eqref{eq:beamsplitter1} under constraint~\eqref{eq:ave_block}, this capacity is $g(\eta \ave)$ nats per channel use \cite{giovannettiguha04}, where 
\begin{equation} \label{eq:g_x} g(x) \triangleq (x+1) \log (x+1) - x \log x,\quad x>0. \end{equation}
This immediately implies that the photon efficiency \eqref{eq:efficiency} satisfies:
\begin{equation}\label{holevo_eff} r_\textrm{quantum}(\eta,\ave) = \frac{g(\eta \ave)}{\eta \ave} = \log{\frac{1}{\eta \ave}} + 1 + o(1). 
\end{equation} 
Note that the efficiency is unbounded, that is,
\begin{equation} \lim_{\ave \downarrow 0} r_\textnormal{quantum}(\eta,\ave) = \infty. \end{equation} 
Hence, in terms of \cite{gallager87,verdu90}, the capacity per unit cost $\sup_{\ave} r(\eta,\ave)$ of the channel \eqref{eq:channel_map} is infinite. 

The capacity $g(\eta \ave)$ is achievable by Alice using product (i.e., nonentangled), pure input states
\begin{equation}\label{eq:indep_input}
|\psi^k \ket = |\psi_1 \ket \otimes |\psi_2 \ket \otimes \cdots \otimes |\psi_k \ket .
\end{equation}
Indeed, in this paper we limit our attention to such mode of operation, where 
the average-photon-number constraint \eqref{eq:ave_block} becomes
\begin{equation}\label{eq:ave}
  \frac{1}{k} \sum_{i=1}^k \bra {\psi_i}|\hat{a}_i^\dag\hat{a}_i| \psi_i \ket \le  \ave.
\end{equation}

For the degenerate case $\eta=1$, a simple capacity-achieving codebook consists only
of number states, where the photon numbers' empirical distribution is independent and identically distributed (i.i.d.) geometric (i.e., Bose-Einstein). Bob's optimal measurement for this codebook is simply per-channel-use direct detection. We shall see in Section~\ref{sec:comm_number} that, in the photon-efficient regime, this code construction can be further simplified and can be used also when $\eta<1$, without sacrificing much photon efficiency.

For the general case where $\eta$ may not be one, the capacity can be achieved if
Alice's codebook consists of coherent states
\begin{equation}
	|\psi^k\ket = |\alpha_1\ket \otimes |\alpha_2\ket \otimes \cdots \otimes |\alpha_k\ket,
\end{equation}
and if Bob performs a general (not per-channel-use) POVM on the output state, which is
\begin{equation}
|\phi^k\ket = |\sqrt{\eta}\,\alpha_1\ket \otimes |\sqrt{\eta}\,\alpha_2\ket \otimes \cdots \otimes |\sqrt{\eta}\,\alpha_k\ket.
\end{equation}
In this case, the average-photon-number constraint \eqref{eq:ave_block} becomes
\begin{equation}\label{eq:ave_coherent}
	\sum_{i=1}^k |\alpha_i|^2 \le k\ave.
\end{equation}
It is known that capacity-achieving codebooks of coherent states should have empirical distributions that resemble i.i.d. complex-Gaussian with mean zero and variance $\ave$ \cite{giovannettiguha04}. The main problem with such a code is that Bob's POVM is almost impossible to implement using today's technology. Hence we are interested in ``practical'' schemes, in particular, in schemes where Bob uses per-channel-use direct detection while Alice sends coherent states. As we shall see in Section~\ref{sec:comm_coherent}, this restriction induces a second-order-term loss in photon efficiency.




\subsection{Alice Sends Binary Number States}\label{sec:comm_number}

Consider the case where the sequence of states sent by Alice consists only of the number states $|0\ket$ and $|1\ket$, and where Bob uses direct detection. Recalling \eqref{eq:number_direct}, for input $|0\ket$ Bob will always detect no photon, while for input $|1\ket$ Bob detects one photon with probability $\eta$, and detects no photon otherwise. Thus the scheme induces a classical Z channel. The maximum achievable rate is, according to the classical channel coding theorem \cite{shannon48}, the maximum mutual information over this channel.

Let
\begin{equation} \label{I_Z} I_\textnormal{Z}(q,\mu) \triangleq H_2(q  \mu) - q H_2(\mu) \end{equation} be the mutual information over a Z channel with input probability $P_X(1)=q$ and transition probability $P_{Y|X}(1|1)=\mu$, where $H_2(\cdot)$ is the binary entropy function
\begin{equation}
	H_2 (x) \triangleq x\log\frac{1}{x}+(1-x)\log\frac{1}{1-x},\quad 0<x<1.
\end{equation} 
Due to the photon-number constraint, the input distribution must satisfy $q \leq \ave$.\footnote{The expected number of photons translates to a per-codeword constraint via a standard expurgation argument.} It is easy to see that $I_\textnormal{Z}(q,\mu)$ is monotonically increasing in $q$ for small enough $q$, and hence, in the regime of interest, we should choose $q=\ave$, achieving rate $I_\textnormal{Z}(\ave,\eta)$. The resulting photon efficiency can be readily shown to satisfy:
\begin{equation}\label{eq:numZ}
 r_\textnormal{num,Z} (\eta,\ave) = \frac{I_\textnormal{Z}(\ave,\eta)}{\eta \ave}  = r_\textrm{quantum}(\eta,\ave) - \frac{H_2(\eta)}{\eta} + o(1), \end{equation} reflecting a constant efficiency loss with respect to the optimum \eqref{holevo_eff}.

For the scheme described above, the task of (classical) coding is difficult: one needs mutual-information-approaching codes for a Z channel with a highly skewed input.  We can solve this problem by replacing the i.i.d. binary codebook by PPM: the input sequence consists of ``frames'' of length $\lceil\nicefrac{1}{\ave}\rceil$, where each frame includes exactly one photon, whose position is uniformly chosen inside the frame. (If the blocklength is not divisible by $\lceil\nicefrac{1}{\ave}\rceil$, then we ignore the remainder.) This scheme converts the channel to a $\lceil\nicefrac{1}{\ave}\rceil$-ary erasure channel. By computing the capacity of this erasure channel, we easily see that the photon efficiency of the PPM scheme is:
\begin{equation} r_\textnormal{num,PPM}(\eta, \ave) = \log {\frac{1}{\ave}} + o(1), \end{equation} which again reflects only a constant loss compared to the optimal efficiency \eqref{holevo_eff}. The large-alphabet erasure channel is much like a packet-erasure channel encountered in internet applications, and good off-the-shelf codes are available.

\subsection{Alice Sends Binary Coherent States}\label{sec:comm_coherent}

Generating the number state $|1\ket$ is hard in practice. We hence turn to coherent states, which are a good model for light coming out of laser sources \cite{shapiro09}. 

We consider a simple binary-coherent-state scheme. In this scheme, Alice first generates a classical binary codebook where the probability of $1$ is $q$. She then maps $0$ and $1$ to the coherent states $|0\ket$ and $\left|\nicefrac{\ave}{q}\right\ket$, respectively. Note that doing this satisfies the average-power constraint \eqref{eq:ave_coherent}. Bob uses direct detection that is not photon-number resolving (PNR), i.e., he views a measurement with no photon as a logical $0$, and views any measurement with at least one photon as a logical~$1$. (Such a detector is easier to build than a PNR detector, which outputs the exact number of detected photons.) This results again in a classical Z channel, with
\begin{equation} P_{Y|X}(1|1) = \mu_\textnormal{coh} (q,\ave) \triangleq 1 - \exp\left(-\frac{\eta \ave}{q}\right). \end{equation} 
We can thus achieve $I_\textnormal{Z}(q,\mu_\textnormal{coh})$ nats per channel use, where $q$ should be chosen to maximize $I_\textnormal{Z}(q,\mu_\textnormal{coh})$. The exact analytical optimization is complicated, but in the photon-efficient regime the approximate optimum (which yields the best rate up to the approximation of interest) is given by
\begin{equation}\label{p_star} q^*(\ave) = \frac{\eta \ave}{2}\log\frac{1}{\ave}. \end{equation}
The resulting photon efficiency is given by: 
\begin{IEEEeqnarray}{rCl}
\label{Z_eff1} r_\textnormal{coh,Z}(\eta,\ave) & = & \frac{I_Z\bigl(q^*(\ave),\mu_\textnormal{coh}(q^*(\ave),\ave)\bigr)}{\eta \ave} \\
	& =&  \log{\frac{1}{\eta \ave}} - \log \log{\frac{1}{\ave}} + \log 2 -1 + o(1). \label{Z_eff} \IEEEeqnarraynumspace
\end{IEEEeqnarray} 

Comparing to the quantum limit \eqref{holevo_eff}, we see that the efficiency loss of the coherent-state-and-direct-detection scheme with respect to the optimal performance grows as $\log \log \nicefrac{1}{\ave}$ as $\ave$ decreases in the photon-efficient regime. This loss is inherent to any ``classical'' transmission scheme, even if general (non-binary) coherent states are sent \cite{wangwornell14}, or if the receiver is allowed to use feedback between measurements \cite{chungguhazheng11}. 

Similarly to the case of Alice sending number states, we can alleviate the difficulty of coding by replacing the i.i.d. codebooks with PPM frames, an idea already exploited in \cite{pierce78,massey81}. Indeed, using PPM frames of length $b$ with the optimum (to the approximation order) choice of~\eqref{p_star} and $b=\lceil \nicefrac{1}{q^*(\ave)} \rceil$, this efficiency is
\begin{equation}\label{PPM_eff}
	r_\textnormal{coh,PPM}(\ave)= \frac{\mu_\textnormal{coh} \bigl(q^*(\ave), \ave\bigr) \log b}{\eta  b \ave},
\end{equation}
and has the same expression as on the right-hand side of \eqref{Z_eff}, i.e., the further efficiency loss incurred by restricting to PPM is $o(1)$. 

  \begin{figure}[t]
        \centering
        \psfrag{E}[cc]{\small $\mathcal{E}$}
	\psfrag{rquantum}[Bl][Bl]{\small $r_\textnormal{quantum}$}
	\psfrag{rz}[Bl][Bl]{\small $r_\textnormal{coh,Z}$}
	\psfrag{rppm}[Bl][Bl]{\small $r_\textnormal{coh,PPM}$}
      \psfrag{Photon efficiency (nats/photon)}[Bl][Bl]{\footnotesize Photon efficiency (nats/photon)} 
        \includegraphics[width=0.5\textwidth]{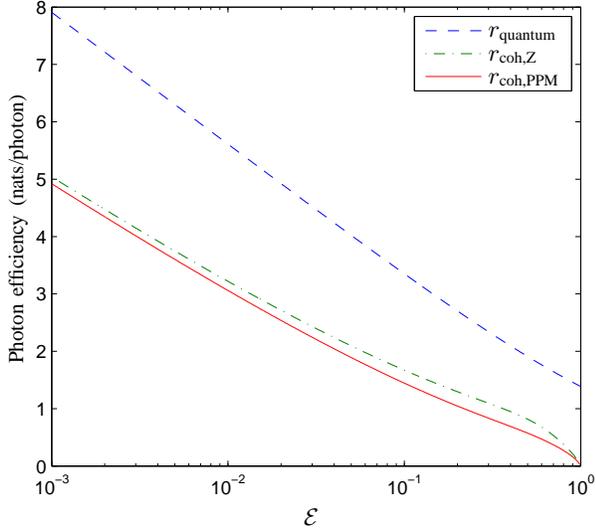}
        \caption{Photon efficiency in the different cases discussed in Section~\ref{sec:communication}. Efficiency in the quantum case $r_\textnormal{quantum}$ is computed from \eqref{holevo_eff}; efficiency for coherent-state inputs and Z-channel model $r_\textnormal{coh,Z}$ from \eqref{Z_eff1}; and efficiency for coherent-state inputs and PPM $r_\textnormal{coh,PPM}$ from \eqref{PPM_eff}. For all three we let the channel be lossless, i.e., we choose $\eta=1$.} 
	\label{fig_eff}  
    \end{figure}

Figure~\ref{fig_eff} depicts the photon efficiency in the different cases discussed in this section. It can be appreciated that, while the loss of using coherent states with direct detection is large, the further loss of PPM is small. As we shall see, similar phenomena are also observed in key-distribution scenarios. 


\section{Key Distribution in Model C}\label{sec:channel_case}

In this section we study the key-distribution problem in Model C, which we set up in Section~\ref{sec:key_channel_setup}. 

To the best of our knowledge, the maximum secret-key rate, and hence also the maximum photon efficiency, in this setting are not yet known.  However, in the photon-efficient regime we have the following asymptotic upper bound. (Later we show that this upper bound is tight within a constant term).

\begin{proposition}\label{prp:C_max}
	The maximum photon efficiency for key distribution in Model~C as described in Section~\ref{sec:key_channel_setup} satisfies
	\begin{equation}\label{eq:PE_channel_max}
		r_\textnormal{max} (\ave) \le \log{\frac{1}{\eta \ave}} + 1 + o(1).
	\end{equation}
\end{proposition}

\begin{proof} We use the fact that the maximum secret-key rate over a quantum channel cannot exceed the communication capacity of the same channel. This follows, e.g., from \cite[Chapter I, Theorem 5.1]{wilmink03}. Recalling \eqref{holevo_eff}, the proof is completed. \end{proof} 

As in the communication setting, we shall mostly focus on key-distribution schemes in which Bob only employs direct detection. As we shall see in Section~\ref{sec:key_channel_number}, if Alice can send number states---even only binary number states---the photon-efficiency loss of direct detection is at most a constant term in the photon-efficient regime. However, in Section~\ref{sec:key_channel_coherent} we show that if Alice can only send coherent states, then the loss in photon efficiency scales like $\log\log\nicefrac{1}{\ave}$. These results are similar to their optical-communication counterparts. Also similar to the communication scenario is the fact that PPM is nearly optimal in terms of photon efficiency; in the context of key distribution, PPM allows us to greatly simplify the coding task in the information-reconciliation step.

\subsection{Alice Sends Binary Number States}\label{sec:key_channel_number}

Consider the following key-distribution scheme. 

\setcounter{sc}{3}

\begin{scheme} \label{Scheme_C-1}
\mbox{}
  \begin{enumerate}
    \item Let $b\triangleq \lceil \nicefrac{1}{\ave} \rceil$. We divide
      the whole block of $k$ channel uses into frames each
      consisting of $b$ consecutive uses (and ignore the remainder).
    \item Alice generates a sequence of integers
      $\tilde{X}_1,\tilde{X}_2,\ldots$ i.i.d.
      uniformly in $\{1,\ldots,b\}$. These are the
      ``pulse positions''. Within the $i$th frame,
      $i\in\{1,2,\ldots\}$, she sends the
      number state $|1\ket$ in the $\tilde{X}_i$th channel use, and
      sends $|0\ket$ in all other channel uses.
    \item Bob makes direct detection on every channel output. Since Alice sends one 
    photon per frame, Bob will either detect a single photon or no photon per frame.
    Let the set of frames where Bob had a detection be denoted as $\{i_1,i_2,\ldots\}$, and denote the detection positions inside
      these bins by $\tilde{Y}_{i_1},\tilde{Y}_{i_2},\ldots$. Bob tells Alice
      the values of $i_1,i_2,\ldots$ using the public channel.
    \item Alice generates the secret key from
      $\tilde{X}_{i_1},\tilde{X}_{i_2},\ldots$, and Bob generates the
      secret key from $\tilde{Y}_{i_1},\tilde{Y}_{i_2},\ldots$, both by
      directly taking the binary representation of these integers.
  \end{enumerate}
\end{scheme}

The average-photon-number constraint \eqref{eq:ave_block} is clearly satisfied. Scheme~\ref{Scheme_C-1} is rather simple in the sense that
\begin{itemize}
	\item Alice's input states are either $|0\ket$ or $|1\ket$;
	\item Bob's detector can be non-PNR;
	\item The information-reconciliation step is uncoded, and only involves one-way communication from Bob to Alice;
	\item There is no privacy-amplification step.
\end{itemize}

As the next proposition shows, this simple scheme performs very well in the photon-efficient regime: it is at most a constant term away from optimum. Compared to the communication case~\eqref{holevo_eff}, this proposition also shows that the loss in photon efficiency due to the secrecy requirement is at most a constant term.

\begin{proposition}\label{prp:C-1}
	Scheme~\ref{Scheme_C-1} generates a secret key between Alice and Bob, and its photon efficiency is
	\begin{equation}
		r_{\textnormal{C-1}}(\eta,\ave) = \log\frac{1}{\ave}+o(1)
	\end{equation}
	for all $\eta\in(0,1]$.
\end{proposition}

\begin{proof}
  We first verify that Scheme~\ref{Scheme_C-1} indeed
  generates a secret key. 
  To this end, first note that
  $\tilde{X}_{i_j}=\tilde{Y}_{i_j}$ for all $j\in\{1,2,\ldots\}$. This
  is because Alice sends only one non-vacuum state in each frame, and
  because Bob cannot detect any photon in a channel use where Alice 
  sends $|0\ket$. Hence the keys obtained by Alice and by Bob are the same. Second, by the way Alice chooses~$\tilde{\vect{X}}$,
  every $\tilde{X}_{i_j}$ (or,
  equivalently, $\tilde{Y}_{i_j}$) is uniformly distributed in
  $\{1,\ldots,b\}$, independently of $\tilde{X}_{i_{j'}}$ where $j'\neq
  j$. This shows that the key is uniformly distributed. It now remains to verify that the key is dependent neither on Eve's
  output states from the optical channel nor on the messages which Bob
  sends to Alice. It is independent of Eve's optical states because,
  in every selected frame, Bob detects the only photon that Alice
  transmits, so Eve's post-measurement state in this frame is the all-vacuum
  state. It is independent of Bob's messages because Bob only sends the labels of the selected frames to Alice, and because Alice chooses the
  pulse positions independently of the frame labels. 

We next compute the photon efficiency achieved by Scheme~\ref{Scheme_C-1}. Let $N(k)$ be the total number of frames selected by Bob within $k$ channel uses. Since each frame is selected when Bob detects a photon in that frame, which happens with probability $\eta$, we have from the Law of Large Numbers that
\begin{equation}
	\lim_{k\to\infty} \frac{N(k)}{k} = \lim_{k\to\infty}\frac{\eta \lfloor\nicefrac{k}{b}\rfloor}{k} = \eta \ave
\end{equation}
with probability one. Each detected photon (or, equivalently, each selected frame)
  provides $\log b$ nats of secret key. So, as $k$ tends to infinity, the achieved photon efficiency tends to
  \begin{equation}
    \lim_{k\to\infty} \frac{N(k)  \log b}{k  \eta \ave } = \log b =
    \log\frac{1}{\ave}+o(1). 
  \end{equation}

\end{proof}

\subsection{Alice Sends Coherent States}\label{sec:key_channel_coherent}

We now restrict Alice to sending coherent states since, as discussed previously, generating the number state $|1\ket$ is hard in practice. Under this restriction, Alice generates a sequence of complex numbers $\alpha_1,\alpha_2,\ldots,\alpha_k$ satisfying \eqref{eq:ave_coherent}, prepares the coherent states
$|\alpha_1\ket,|\alpha_2\ket,\ldots,|\alpha_k\ket$, and sends them over the
channel. As the next proposition shows, this restriction induces a loss of $\log\log\nicefrac{1}{\ave}$ in the photon efficiency, even if the scheme employed is more sophisticated than Scheme~\ref{Scheme_C-1}.

\begin{proposition}\label{prp:coherent_converse}
  The maximum photon efficiency in Model~C when Alice sends only
  coherent states and when Bob uses only direct detection satisfies
  \begin{equation}\label{eq:coherent}
    r_{\textnormal{coh}} (\eta,\ave) \le \log\frac{1}{\ave} -
    \log\log\frac{1}{\ave} + O(1)
  \end{equation}
  for all $\eta\in(0,1]$. 
\end{proposition}

\begin{proof}
  We note  that, when Alice sends the coherent state $|\alpha\ket$, Bob's measurement
  outcome $Y$ has a Poisson distribution of mean
  $\eta|\alpha|^2$. We can bound the achievable secret-key rate as
  \begin{IEEEeqnarray}{rCl}
    R_\textnormal{coh} (\eta,\ave) & \le & \max_{\E{|X|^2} \le \ave} I(X;Y) \label{eq:coherent_2}\\
    & = & \max_{\E{|X|^2} \le \ave} I(|X|^2;Y), \label{eq:coherent_1}
  \end{IEEEeqnarray}
  where \eqref{eq:coherent_2} follows because the secret-key rate over a channel cannot be larger than the communication capacity of the channel (see, e.g., \cite{ahlswedecsiszar93}); and where \eqref{eq:coherent_1} follows because $|X|^2$ is a deterministic function of $X$, and because $X\markov |X|^2 \markov Y$
  forms a Markov chain. Finally, the right-hand side of \eqref{eq:coherent_1}, which is the maximum mutual information over a Poisson channel under an average-photon-number constraint, is shown in \cite{wangwornell14} to satisfy
  \begin{equation} \label{eq:from_seminal}
    \max_{\E{|X|^2} \le \ave} I(|X|^2;Y) \le \eta \ave \left\{ \log\frac{1}{\ave} -
      \log\log\frac{1}{\ave} +O(1) \right\}.
  \end{equation}
\end{proof}

We do not specify the $O(1)$ term, as the derivation of \eqref{eq:from_seminal} in\cite{wangwornell14} yields expressions that are rather involved. In the sequel we show that the bound~\eqref{eq:coherent} is tight within a constant term.

As in Section~\ref{sec:comm_coherent}, to simplify the coding task for the information-reconciliation step, Alice and Bob can use a PPM-based scheme. We choose the PPM frame-length to be:
      \begin{equation}\label{eq:blocklog}
        b\triangleq \left\lceil \frac{1}{\ave\log\nicefrac{1}{\ave}}
        \right\rceil.
      \end{equation}
This choice is optimal up to the order of approximation of interest. Note that $b$ in \eqref{eq:blocklog} is half the frame-length chosen for the communication setting, where the latter is $\lceil\nicefrac{1}{q^*(\ave)}\rceil$ with $q^*(\ave)$ given in~\eqref{p_star}.

\begin{scheme}\label{Scheme_C-2}
\mbox{}
  \begin{enumerate}
    \item 
      We divide
      the whole block of $k$ channel uses into frames each
      consisting of $b$ consecutive uses (and ignore the remainder).
    \item Alice generates a sequence of integers
      $\tilde{X}_1,\tilde{X}_2,\ldots$ i.i.d.
      uniformly in $\{1,\ldots,b\}$. Within the $i$th frame,
      $i\in\{1,2,\ldots\}$, she sends the
      coherent state $|\sqrt{b \ave}\ket$ in the $\tilde{X}_i$th
      channel use, 
      and sends the vacuum state 
      $|0\ket$ in all other channel uses.
    \item Bob makes direct detection on every channel-output. Since all channel input-states but one are in vacuum state, he will have detections in at most one output. Let the set of frames where Bob had a detection be denoted as $\{i_1,i_2,\ldots\}$, and denote the detection positions inside
      these bins by $\tilde{Y}_{i_1},\tilde{Y}_{i_2},\ldots$.  He tells Alice
      the values of $i_1,i_2,\ldots$ using the public channel.
    \item Alice generates the raw key $K_A$ from
      $\tilde{X}_{i_1},\tilde{X}_{i_2},\ldots$, and Bob generates the
      raw key $K_B$ from $\tilde{Y}_{i_1},\tilde{Y}_{i_2},\ldots$, both by
      directly taking the binary representation of these integers.
    \item Alice and Bob perform privacy amplification on their raw keys to
      obtain the secret keys.
  \end{enumerate}
\end{scheme}

The average-photon-number constraint \eqref{eq:ave_block} or \eqref{eq:ave_coherent} is clearly satisfied. Also note that, in this scheme,
  \begin{itemize}
    \item Alice's input states are binary: either $|0\ket$ or $|\sqrt{b \ave}\ket$;
    \item Bob's detector can be non-PNR;
    \item The information-reconciliation step is uncoded, and only involves one-way
      communication from Bob to Alice.
  \end{itemize}
In contrast to the restriction on Alice to sending only coherent states, which results in a loss of $\log\log \nicefrac{1}{\ave}$ in photon efficiency, the further simplifications employed in Scheme~\ref{Scheme_C-2} induce at most a constant-term loss.

\begin{proposition}\label{prp:C-2}
	Scheme~\ref{Scheme_C-2} achieves photon efficiency
	\begin{equation}
		r_{\textnormal{C-2}} (\eta,\ave) \ge \log\frac{1}{\ave} - \log\log\frac{1}{\ave} - (1-\eta) + o(1)
	\end{equation}
	for all $\eta\in(0,1]$.
\end{proposition}

The proof, which appears in Appendix~\ref{app:C-2}, is more involved than that of Scheme~\ref{Scheme_C-1}, since in the case of coherent states, the raw key depends upon Eve's optical states (since, if Bob and Eve both see detections in some frame, then they must be in the same location). However, we bound the information leakage and show that it leads to at most a constant key-efficiency loss. 


\section{Key Distribution in Model S}\label{sec:source_case}

In this section we study the key-distribution problem in Model S, which we set up in Section~\ref{sec:key_source_setup}. Apart from Section~\ref{sec:etaA}, we shall focus on the case where $\eta_A=1$. In this case, we omit the subscript of $\eta_B$ to denote it simply as~$\eta$. 


\begin{proposition}\label{prp:IAB}
	The maximum photon efficiency achievable in Model~S satisfies
	\begin{equation} \label{eq:efficiency_Squantum}
	r_\textnormal{quantum}(\eta,\ave) \le \log{\frac{1}{\eta \ave}} + 1 + o(1).
	\end{equation}
\end{proposition}

\begin{proof}
We note that, without further constraints, the secret-key rate and hence the photon efficiency achievable in Model~S cannot exceed those achievable in Model~C. This is because any measurement Alice performs in Step~1' in Model~S, which is described in Section~\ref{sec:key_source_setup}, can be simulated in Model~C in the following way. Alice first generates random numbers that have the same statistics as the outcomes of the measurement that she would perform in Model~S. Then, for each number, she generates the corresponding post-measurement state on $\mathbb{D}$ and sends it to Bob. Doing these will generate the same correlation between Alice, Bob, and Eve as the corresponding strategy in Model~S would do. The claim now follows immediately from Proposition~\ref{prp:C_max}.
\end{proof}

\emph{Note:} The above proof says that, when Alice and Bob can both use fully quantum devices, there is no advantage in Model~S over Model~C. However, as we later show, this need not be the case when Alice and Bob are restricted, e.g, to direct detection.

For practicality, for the rest of this section we restrict both Alice and Bob to using only direct detection on their quantum states. In fact, Alice and Bob will only use non-PNR direct detection. In contrast, we do not impose any constraint on Eve's measurement, thus our schemes are secure against a fully-quantum (though passive) Eve.

\subsection{Direct Detection Combined with Optimal Binary Slepian-Wolf Codes}

After Alice and Bob perform direct detection on their optical states, each of them has a binary sequence where $1$ indicates photons are detected in the corresponding source use. Denote their sequences by $\vect{A}$ and $\vect{B}$, respectively. Due to our source model, $\vect{A}$ and $\vect{B}$ are distributed i.i.d. in time, while each pair $(A,B)$ has joint distribution according to a Z channel with 
\begin{subequations}\label{eq:PAB}
\begin{IEEEeqnarray}{rCl}
	q &\triangleq& P_A(1)  =  1 - \e^{-\ave}\\
	\mu & \triangleq& P_{B|A}(1|1)  =  \frac{1-\e^{-\eta \ave}}{1-\e^{-\ave}}.
\end{IEEEeqnarray}
\end{subequations}
Bob can help Alice to know $\vect{B}$ by sending her a Slepian-Wolf code \cite{slepianwolf73}. For the moment, we assume that Alice and Bob have an optimal Slepian-Wolf code for the joint distribution $P_{AB}$ (Later we drop this assumption to find more realistic code constructions.) Then they can use the following key-distribution scheme.

\setcounter{sc}{19}
\setcounter{scheme}{0}

\begin{scheme}\label{Scheme_S-1}
\mbox{}
\begin{enumerate}
	    \item Alice and Bob perform non-PNR direct detection to obtain binary sequences $\mathbf{A}$ and $\mathbf{B}$, respectively.
	\item Bob sends Alice an optimal Slepian-Wolf code so that Alice knows $\vect{B}$ with high probability. They use $\vect{B}$ as the raw key.
	\item Alice and Bob perform privacy amplification on $\vect{B}$ to obtain the secret key.
\end{enumerate}
\end{scheme}

The key rate and photon efficiency of Scheme~\ref{Scheme_S-1} satisfy the following.
\begin{proposition}\label{prp:S-1}
	Scheme~\ref{Scheme_S-1} achieves the key rate
	\begin{equation}\label{eq:S-1_IAB}
		R_\textnormal{S-1}(\eta,\ave) = I(A;B)
	\end{equation}
	where the mutual information is computed on the joint distribution $P_{AB}$ given by \eqref{eq:PAB}.
	Furthermore, for all $\eta\in(0,1]$, the photon efficiency of Scheme~\ref{Scheme_S-1} satisfies
	\begin{equation}\label{eq:S-1_efficiency}
		r_\textnormal{S-1}(\eta,\ave) = \log\frac{1}{\eta \ave} + 1 - \frac{H_2(\eta)}{\eta} + o(1).
	\end{equation}
\end{proposition}
\begin{proof}
We first prove \eqref{eq:S-1_IAB}. Its converse part follows immediately from \cite[Chapter I, Theorem 5.3]{wilmink03}, which states that the secret-key rate cannot exceed $I(A;B)$ even if Eve possesses no quantum state that is correlated to $A$ and $B$. Its achievability part follows from \cite[Chapter III, Theorem 2.2]{wilmink03}: when we eliminate the ``helper subalgebra'', the theorem says that the forward key capacity (i.e., the maximum key rate achievable when Alice does not communicate to Bob) is lower-bounded by $I(A;B) - I(B;\mathbb{E})$ evaluated for the joint state consisting of Alice's and Bob's measurement outcomes and Eve's post-measurement state. As shown in Section~\ref{sec:key_source_setup}, Bob's measurement outcome is independent of Eve's post-measurement state, so $I(B;\mathbb{E})=0$.\footnote{In Section~\ref{sec:key_source_setup} we consider the case where Bob performs a complete projective measurement in the number-state basis, whereas here Bob's non-PNR detection only distinguishes between zero and positive photon numbers. But extending our claim for the former case to the latter is straightforward.} 

We next prove \eqref{eq:S-1_efficiency}. Direct evaluation for the Z-channel mutual information \eqref{I_Z} for the channel parameters $q$ and $\mu$ of \eqref{eq:PAB} gives:
\begin{IEEEeqnarray}{rCl} 
	I(A;B) & = & I_Z(q,\mu) \\
	& = & H_2(\e^{-\eta \ave}) - \left(1-\e^{-\ave}\right) H_2\left(\frac{1-\e^{-\eta \ave}}{1-\e^{-\ave}}\right) \IEEEeqnarraynumspace\\
	& = & \eta \ave\log\frac{1}{\eta \ave} + \eta \ave-\ave H_2(\eta)+o(\ave). \label{eq:S-1_last} 
\end{IEEEeqnarray}
Substituting in  \eqref{eq:S-1_IAB} and dividing by $\eta \ave$ yields \eqref{eq:S-1_efficiency}.
\end{proof}


Hence the conceptually simple Scheme~\ref{Scheme_S-1}, which only uses non-PNR direct detection both at Alice and at Bob, is at most a constant term away from the optimal quantum efficiency whose upper bound is given in \eqref{eq:efficiency_Squantum} . Comparing this with \eqref{holevo_eff} and \eqref{eq:PE_channel_max} we see that the differences between the optimal photon efficiencies in communication, in Model~C, and in Model~S are at most constants. Interestingly, $r_\textnormal{S-1}(\eta,\ave)$ is asymptotically the same as the photon efficiency in the communication scenario where Alice sends binary number states \eqref{eq:numZ}. 

The problem with Scheme~\ref{Scheme_S-1} is, though, that the source distribution $P_{AB}$ is highly skewed, which makes it difficult to find a good Slepian-Wolf code, much like the difficulty to obtain a channel code in the communication setting of Section~\ref{sec:communication}. While in communication and in Model~C Alice can use PPM to simplify code design, in Model~S this is no longer possible, as the sequences $\vect{A}$ and $\vect{B}$ are governed by the source, over which neither Alice nor Bob have control. Nevertheless, Alice and Bob can use a PPM-like scheme by \emph{parsing} the sequences into frames, as we next propose.

\subsection{Simple Frame-Parsing}

In a simple PPM-like scheme, Alice and Bob parse the source uses into frames, and only use the frames where each of them has exactly one detection to generate the key.

\begin{scheme}\label{Scheme_S-2}
\mbox{}
\begin{enumerate}
	    \item Alice and Bob perform non-PNR direct detection to obtain binary sequences $\mathbf{A}$ and $\mathbf{B}$, respectively.
	\item Let $b$ be as in \eqref{eq:blocklog}. We divide the whole block of $k$ source uses into frames each consisting of $b$ consecutive uses (and ignore the remainder).
	\item Bob selects all the frames in which he detects at least one photon ($B=1$ for at least one source use). Denote the labels of these frames by $\{i_1,i_2,\ldots\}$. He tells Alice the values of $i_1,i_2,\ldots$ using the public channel.
	\item Alice selects the frames among $i_1,i_2,\ldots$ in which $A=1$ for \emph{exactly one} source use. Denote the labels of these frames by $\{i_{j_1},i_{j_2}\ldots\}$, Alice's detection positions within these frames by $\{Y_{i_{j_1}},Y_{i_{j_2}},\ldots\}$, and Bob's (unique) detection positions within these frames by $\{X_{i_{j_1}},X_{i_{j_2}},\ldots\}$. She tells Bob the values of $j_1,j_2,\ldots$ using the public channel.
	\item Alice and Bob generate the raw key by taking the binary representations of $\{X_{i_{j_1}},X_{i_{j_2}},\ldots\}$ and of $\{Y_{i_{j_1}},Y_{i_{j_2}},\ldots\}$, respectively.
	\item Alice and Bob perform privacy amplification on the raw key to obtain the secret key.
\end{enumerate}
\end{scheme}

As in Schemes~\ref{Scheme_C-1} and \ref{Scheme_C-2}, the information-reconciliation step in Scheme~\ref{Scheme_S-2} is uncoded and hence very simple. The performance of Scheme~\ref{Scheme_S-2} is similar to that of Scheme~\ref{Scheme_C-2} where Alice sends coherent states, in the sense that it loses a $\log\log\nicefrac{1}{\ave}$ term in photon efficiency compared to the optimum \eqref{eq:efficiency_Squantum}. Interestingly, here the loss does not come from the input states used, as they are identical to those in Scheme~\ref{Scheme_C-1}, but rather from the parsing process.

\begin{proposition}\label{prp:S-2}
	The photon efficiency of Scheme~\ref{Scheme_S-2} satisfies
	\begin{equation}\label{eq:S-2}
	 	r_{\textnormal{S-2}}(\eta,\ave) = \log\frac{1}{\ave} - \log\log\frac{1}{\ave} - 1 + o(1).
	\end{equation}
\end{proposition}

The scheme has some information leakage, since Eve can use her knowledge about the frames which were selected for key generation (obtained by listening to the public channel), in conjunction with the measurements she performs on the same frames. The proof, which appears in Appendix~\ref{app:S-2}, shows that this leakage is vanishing in the photon-efficient limit.

\emph{Note:} If Alice uses PNR direct detection (which is technically more difficult than non-PNR), then Scheme~\ref{Scheme_S-2} can be simplified so that it does not contain an privacy-amplification step. Indeed, Alice can select those frames in which she detects \emph{only one photon}. In this case, since Bob also detects photons (in fact, only one photon) in every such frame, we know that Eve's post-measurement states in these frames are all vacuum. Hence Eve has no information about $\tilde{X}$, and taking the binary representation of $\tilde{X}$ already gives Alice and Bob a secret key.

The information loss of Scheme~\ref{Scheme_S-2} compared to Scheme~{\ref{Scheme_S-1}} comes from two sources. First, the sequence $\{i_1,i_2,\ldots\}$ itself contains useful information that can be used to generate secret bits, but is not exploited in Scheme~\ref{Scheme_S-2}. Second, frames in which Alice detects photons in two or more source uses are discarded. As it turns out, the first source of information loss is dominant in the photon-efficient regime; we next show how this loss can be recovered. (Loss from the second source can also be partially recovered, e.g., by varying the frame-lengths~\cite{zhouwornell13}.)

\subsection{Enhanced Frame-Parsing}

Our idea of enhancing the frame-parsing scheme~\ref{Scheme_S-2} is to extract secret-key bits also from the sequence $\{i_1,i_2,\ldots\}$, which indicates the positions of frames selected by Bob. To this end, instead of sending this sequence uncoded, Bob uses a binary Slepian-Wolf code to send this information to Alice. Note that such a code is much easier to construct than the one in Scheme~\ref{Scheme_S-1}, as the zeros (frames not selected by Bob) and ones (frames selected by Bob) are much more balanced than in the original binary sequence $\vect{B}$; recall \eqref{eq:blocklog}.  Assuming that an optimal Slepian-Wolf can be found, we can completely recover the $\log\log\nicefrac{1}{\ave}$ term and reduce the loss in photon efficiency to a constant term.

\begin{scheme}\label{Scheme_S-3}
\mbox{}
\begin{enumerate}
	    \item Alice and Bob use non-PNR direct detection to obtain binary sequences $\mathbf{A}$ and $\mathbf{B}$, respectively.
	\item Let $b$ be as in \eqref{eq:blocklog}. We divide the whole block of $k$ source uses into frames each consisting of $b$ consecutive uses (and ignore the remainder).
	\item Let $\tilde{B}_i$ be the indicator that Bob detects at least one photon within the $i$th frame, and let $\tilde{A}_i$ be the same indicator for Alice. Bob sends a Slepian-Wolf code to Alice using the public channel, so that Alice can recover $\tilde{\vect{B}}$ based on the codeword together with $\vect{\tilde{A}}$ with high probability.
	\item Corresponding to every $i$ such that $\tilde{B}_i=1$, Alice sends a binary symbol $C_i$ to Bob: $C_i=1$ if within the $i$th frame there is \emph{exactly} one source use where $A=1$, and $C_i=0$ otherwise. Note that since Alice knows $\vect{\tilde{B}}$ with high probability, she can send $C_i$s simply as a bitstream in an increasing order in $i$ (and skip the $i$s for which $\tilde{B}_i=0$).
	\item Alice and Bob perform privacy amplification on $\vect{\tilde{B}}$ to obtain the first part of the secret key.
	\item For every $i$ such that $\tilde{B}_i=C_i=1$, let $X_i$ be the position where $A=1$, and let $Y_i$ be the (unique) position where $B=1$.  Alice and Bob generate the second part of the secret key by taking the binary representations of $X_i$ and of $Y_i$, respectively, for all such $i$s, and by then performing privacy amplification.
\end{enumerate}
\end{scheme}

\begin{proposition}\label{prp:S-3}
	Scheme~\ref{Scheme_S-3} achieves photon efficiency
	\begin{equation}\label{eq:S-3}
		r_\textnormal{S-3} (\eta,\ave) \ge \log\frac{1}{\ave}-\frac{H_2(\eta)}{\eta}+o(1)
	\end{equation}
	for all $\eta\in(0,1]$.
\end{proposition}

The proof, which appears in Appendix~\ref{app:S-3}, evaluates the key rate that  Step 5) adds over the rate of Scheme~\ref{Scheme_S-2}. This part of the key consists of frame labels, thus it is obviously correlated with the messages sent over the public channel. However, we show that in the photon-efficient limit Eve must ``lose synchronization'' with the frame locations, thus the leakage is vanishing.

\subsection{Extension to the Case $\eta_A<1$}\label{sec:etaA}
The results for the case where $\eta_A$ in \eqref{eq:beamsplitterA} is equal to one can be extended to the case where $\eta_A<1$, though the expressions become considerably more cumbersome. We hence only give some heuristic explanations how our schemes should be modified, and how they perform. Note that for the following discussions the photon efficiency is defined in \eqref{eq:PE_etaAB}. Also recall that we assume that the source is co-located with Alice, such that the photons lost do not reach Eve.

\emph{Quantum Limit:} Proposition~\ref{prp:IAB} holds  but with a different constant term. The same proof ideas apply.

\emph{Direct Detection:} Scheme~\ref{Scheme_S-1} can be directly applied to the case where $\eta_A<1$ without modification, and its photon efficiency is different from the right-hand side of \eqref{eq:S-1_efficiency} by a constant term, i.e., it is again at most a constant away from the quantum limit.

\emph{Simple Frame-Parsing:} Scheme~\ref{Scheme_S-2} needs some modifications in order to work when $\eta_A<1$. First, in Step 2) Bob should select only those frames in which there is \emph{exactly} one source use where $B=1$. This is because there can be frames in which Bob has more detections than Alice, due to the loss to Alice. Second, after Step~3) Bob needs to send Alice a $b$-ary Slepian-Wolf code on his detection positions inside the selected frames, so that Alice will know these positions with high probability. (This is a large-alphabet code for symmetric errors, and is relatively easy to construct.) This step is needed because, since both Alice and Bob only observe lossy versions of the source, their detection positions inside the selected frames might be different. Indeed, the two positions are equal if they come from the same source photon-pair, and are independent of each other if they come from two different source photon-pairs. Finally, for Step~5) (privacy amplification), Eve's side information needs to be examined more carefully compared to the case where $\eta_A=1$. After these modifications, one can show that the photon efficiency is the same as the right-hand side of \eqref{eq:S-2} up to the second term, i.e., the loss in photon efficiency scales like $\log\log\nicefrac{1}{\ave}$.

\emph{Enhanced Frame-Parsing:} If we incorporate the aforementioned modifications for Scheme~\ref{Scheme_S-2} to Scheme~\ref{Scheme_S-3}, then Scheme~\ref{Scheme_S-3} also works for the case $\eta_A<1$, and its photon efficiency is different from the right-hand side of \eqref{eq:S-3} by a constant.

We finally note that, for all three cases in which we restrict Alice and Bob to using direct detection, we can also take \emph{detector dark counts} into account. Statistically, a dark count at Alice can be treated as a source photon-pair that reaches Alice but not Bob; similarly for a dark count at Bob. For example, when the dark-count rates at Alice and at Bob are $\lambda_A$ and $\lambda_B$ counts per slot, respectively, we can model the system by replacing $\eta_A$, $\eta_B$, and $\ave$ with $\eta_A'$, $\eta_B'$, and $\ave'$ that are solved from
\begin{subequations}
\begin{IEEEeqnarray}{rCl}
	\eta_A' \ave' & = & \eta_A \ave + \lambda_A\\
	\eta_B' \ave' & = & \eta_B \ave + \lambda_B\\
	\eta_A' \eta_B' \ave' & = & \eta_A \eta_B \ave
\end{IEEEeqnarray}
\end{subequations}
without introducing any new elements to the model. Note that this replacement of parameters yields the desired correlation between Alice's and Bob's photon counts, but does \emph{not} yield the correct form for Eve's optical states after Alice's and Bob's measurements. However, as our proofs show, information in Eve's optical states does not affect the dominant terms in secret-key rate in the regime of interest.
This observation combined with our results shows that dark counts only affect the constant term in photon efficiency, which is again similar to the previous results in optical communications~\cite{wangwornell14}.


\section{Discussion: Toward Secrecy with a General Adversary}\label{sec:discussion}

In this work we have presented schemes that approach the optimal key rate in the photon-efficient limit, up to a constant efficiency loss. Moreover, these schemes are practical, both in the physical sense (utilizing realizable transmissions and measurements) and in the algorithmic sense (using simple protocols and off-the-shelf codes). However, throughout the work we have assumed that Eve is limited to passive eavesdropping through a beamsplitter channel. We now comment on the problems that may arise when this model does not hold, and point out ways to overcome them. 

First, suppose that Eve is still passive, but is free to change the beamsplitter transmissivity $\eta$ as a function of time, as long as it satisfies some average constraint $\bar \eta$. We now distinguish between two strategies that Eve can use:
\begin{enumerate}
\item Pre-scheduled transmissivity. Take, for example, Scheme~\ref{Scheme_C-2}, and imagine that for each PPM frame, Eve uses $\eta=0$ for half the block, and $\eta=1$ for the other half. Then she knows that the key pertaining to this frame must correspond to the part where $\eta=1$, gaining one bit per detected photon (thus reducing the key efficiency by $\log 2$). This kind of attack can go undetected, provided that Eve randomizes the schedule. However, it is plausible that the efficiency loss is bounded by a constant for any schedule.
\item Measurement-dependent transmissivity. In principle, Eve can change $\eta$ in a causal manner, based upon her measurement outcomes. However, we believe that the gain from using measurements can be shown to vanish in the photon-efficient limit, by the same techniques used to show that the information leakage is small.
\end{enumerate}
It however still remains to be investigated whether our intuitions above are correct, i.e., whether Eve indeed cannot gain from changing the beamsplitter transmissivity.

If Eve is allowed to transmit as well, other types of attacks are possible. A very simple and efficient one is ``intercept and resend'': Eve uses direct detection on the channel meant for Bob, and then upon detection of a photon, transmits a substitute one to Bob. This way Eve can obtain information about Bob's sequence of detections, and if she uses much higher bandwidth than Bob, the delay will not be detected. 

In fact, all QKD protocols face this problem. For example, in the BB84 protocol \cite{bennettbrassard84}, the key is generated using the polarization of a photon; Eve can make a measurement, then transmit to Bob a photon with the same polarization. The solution for BB84 is that Alice and Bob measure in either of two \emph{mutually unbiased} bases, according to local randomness. Only if they happened to measure in the same basis, the measurement results are used, inflicting a rate loss of factor~$2$. By sacrificing rate, they can now \emph{a posteriori} find out whether they used the same basis, and compare the correlation of the polarizations to the expected statistics, thus authenticating the received photons.

Extending this idea to schemes based on photon arrival times involves an extension of the concept of mutually unbiased bases to continuous variables; see \cite{weigertwilkinson08}. 
Specifically, in Model C the modulation and measurements can be performed either in the time or in the frequency domain with the help of dispersive optics; see \cite{mowerzhang12}. Alternatively, in Model S, one can use interferometry to verify that the photons received by Alice and Bob are indeed entangled; see \cite{mowerwong11}.


\appendix

\subsection{Proof of Proposition~\ref{prp:C-2}}\label{app:C-2}
  By the same argument as in the proof of
  Proposition~\ref{prp:C-1}, we know that the raw keys generated by Alice and Bob (before privacy amplification) are the same, and are independent of Bob's messages in the information-reconciliation step. It is, however, dependent on Eve's optical states. We thus need to determine how much secret key can be distilled from the raw key.
  
  The quantum states in different frames are mutually independent, so we need only to analyze one frame that is selected by Bob. We note that, when 
  Alice sends the coherent state $|\sqrt{b \ave}\ket$, 
  Eve's output
  state is $|\sqrt{(1-\eta)b \ave}\ket$, and is independent of
    Bob's measurement outcome conditional on Alice's input. Thus, using~\eqref{eq:rennerkonig}, we know
    that the number of secret nats we can obtain in each selected frame can
    be arbitrarily close to $H(\tilde{X}| \rho^{\mathbb{E}^b})$, where $\tilde{X}$ is uniformly distributed over $\{1,\ldots,b\}$, and where
    $\rho^{\mathbb{E}^b}$ is a $b$-mode bosonic state described as follows: conditional on
    $\tilde{X}=i$, $i\in\{1,\ldots,b\}$, $\rho^{\mathbb{E}^b}$
    has the
    coherent state $|\sqrt{(1-\eta)b \ave}\ket$ in the $i$th
    mode and has the vacuum state $|0\ket$ in all other modes. Note
    that  
    the total number of photons in $\rho^{\mathbb{E}^b}$ is $(1-\eta)
    b \ave$, so
    \begin{IEEEeqnarray}{rCl}
      H(\rho^{\mathbb{E}^b}) & \le & b \big\{\bigl(1+(1-\eta)\ave\bigr)
      \log \bigl(1+(1-\eta)\ave\bigr) \nonumber\\
	& & ~~~{}- (1-\eta)\ave\log
      \bigl((1-\eta)\ave\bigr)\bigr\}\label{eq:coherent_22}\\
      & = & \left\lceil \frac{1}{\ave\log\nicefrac{1}{\ave}}
        \right\rceil \left\{(1-\eta)\ave\log\frac{1}{\ave} + O(\ave) \right\}
      \\
      & = & (1-\eta)+o(1).
    \end{IEEEeqnarray}
    Here, \eqref{eq:coherent_22} follows from the well-known fact that
    the maximum entropy 
    of a $b$-mode bosonic state with a certain average photon number is achieved by the state consisting of $b$ i.i.d. thermal states
    \cite{holevosohmahirota99}. Now the number of secret nats per 
    selected frame satisfies
    \begin{IEEEeqnarray}{rCl}
      H(\tilde{X}| \rho^{\mathbb{E}^b}) & = & H(\tilde{X}) - I(\tilde{X};\rho^{\mathbb{E}^b}) \\
      & \ge & H(\tilde{X}) - H(\rho^{\mathbb{E}^b}) \\
      & = & \log b - H(\rho^{\mathbb{E}^b})\\
      & = & \log\frac{1}{\ave} - \log\log\frac{1}{\ave} - (1-\eta) + o(1).
    \end{IEEEeqnarray}

    We next consider the number of frames per $k$ channel uses that will be selected by Bob, which we denote by $N(k)$. When Alice sends $|\sqrt{b \ave}\ket$, Bob's output has
    a Poisson distribution of mean $\eta  b \ave$, so the
    probability that Bob detects at least one photon is $1-\e^{-\eta 
    b \ave}$. Hence, by the Law of Large Numbers,
    \begin{equation}\label{eq:C-2_1}
	\lim_{k\to\infty} \frac{N(k)}{k} =\lim_{k\to\infty} \frac{(1-\e^{-\eta  b \ave})\lceil \nicefrac{k}{b}\rceil }{k} = \frac{1-\e^{-\eta  b \ave}}{b}
    \end{equation}
    with probability one. Using
\begin{equation}
	\e^{-x} \le 1 - x + \frac{x^2}{2},\quad x\ge 0,
\end{equation}
the right-hand side of \eqref{eq:C-2_1} can be lower-bounded as
    \begin{equation}
      \frac{1-\e^{-\eta  b \ave}}{b} \ge \eta \ave \left(1 -
        \frac{\eta  b   \ave}{2}\right).
    \end{equation}
    The photon efficiency of the proposed scheme can now be
    lower-bounded as
    \begin{IEEEeqnarray}{rCl}
      r_{\textnormal{C-2}} & = & \frac{1}{\eta \ave} \cdot \frac{1-\e^{-\eta  b \ave}}{b} \cdot H(\tilde{X}| \rho^{\mathbb{E}^b})\\
      & \ge & \left(1 - \frac{\eta  b \ave}{2}\right)
      \left\{\log\frac{1}{\ave} - \log\log\frac{1}{\ave} 
        - (1-\eta) + o(1)\right\} \nonumber\\ {} \\
     & = & \left(1-\frac{\eta \left\lceil \frac{1}{\ave \log\nicefrac{1}{\ave}} \right\rceil  \ave}{2}  \right) \nonumber\\
	& & {}\cdot \left\{\log\frac{1}{\ave} - \log\log\frac{1}{\ave} 
        - (1-\eta) + o(1)\right\} \\
      & = & \log\frac{1}{\ave} - \log\log\frac{1}{\ave} - (1-\eta) + o(1),
    \end{IEEEeqnarray}
    which is as claimed.



\subsection{Proof of Proposition~\ref{prp:S-2}}\label{app:S-2}

We first observe that, in every selected frame, the detection positions of Alice and Bob must be the same. This is because, due to \eqref{eq:PAB}, $B=1$ can happen only if $A=1$, and because by our choice each selected frame contains only one source use where $A=1$. We thus know that Alice's and Bob's raw keys are the same with probability one. 

To obtain the secret-key rate, we need to compute the entropy of the raw key conditional on Eve's observations. Note that the quantum states inside different frames are mutually independent. We consider one frame that is selected by Alice and Bob. Denote the detection position in the frame by $\tilde{X}$. It is clear that $\tilde{X}$ is uniformly distributed over $\{1,\ldots,b\}$ and is independent of the label of this frame. All Eve's information about $\tilde{X}$ is in her optical state from the $b$ source uses that form this frame: if the source use where $A=B=1$ contains more than one photons, then Eve could also detect a photon in this source use, hence knowing Alice's and Bob's detection position. But, as we next show, this information leakage is small. To this end, we first note that in source uses where $A=B=0$, Eve's optical state is vacuum. Indeed, according to our source model, the number of photons in Alice's state equals the sum of the numbers of photons in Bob's and Eve's states with probability one. Therefore, when both Alice and Bob make direct detections on a source use and observe no photon, Eve's post-measurement state in the same source use becomes the vacuum state. In the (unique) source use where $A=B=1$, Eve's post-measurement state is the same as her state \emph{without} the condition $A=B=1$ given in \eqref{eq:sigmaE}. This is because $A=B=1$ means nothing but that Bob's photon number is positive, but Eve's post-measurement state is independent of Bob's photon number, as shown in Section~\ref{sec:key_source_setup}. Denote Eve's state over the whole frame by $\sigma^{\mathbb{E}^b}$. We now know that it consists of $b-1$ vacuum states and one state of the form \eqref{eq:sigmaE} whose position inside the frame is random. The expected number of photons in $\sigma^{\mathbb{E}^b}$ is $(1-\eta)\ave$, so the entropy of $\sigma^{\mathbb{E}^b}$ is upper-bounded by \cite{holevosohmahirota99}
\begin{equation}
	H(\sigma^{\mathbb{E}^b}) \le b\cdot g\left(\frac{(1-\eta)\ave}{b}\right) = o(1).
\end{equation}
Thus the amount of secret information extractable from one selected frame is lower-bounded by
\begin{IEEEeqnarray}{rCl}
	H(\tilde{X}| \sigma^{\mathbb{E}^b}) & = & H(\tilde{X}) - I(\tilde{X};\sigma^{\mathbb{E}^b})\\
	& \ge & H(\tilde{X}) - H(\sigma^{\mathbb{E}^b})\\
	& \ge & \log \left\lceil \frac{1}{\ave\log\nicefrac{1}{\ave}} \right\rceil +o(1)\\
	& = & \log\frac{1}{\ave} - \log\log\frac{1}{\ave} + o(1).
\end{IEEEeqnarray}

It now remains to compute the probability that a specific frame will be selected by Alice and Bob. A simple lower bound on the probability of a frame being selected is the following: suppose both Bob and Eve make PNR direct detections on their states, then a frame is selected by Alice and Bob if (but not only if) Bob detects exactly one photon in the frame while Eve detects no photon. Bob's photon number has a Poisson distribution of mean $\eta  b \ave$, while Eve's photon number has a Poisson distribution of mean $b(1-\eta)\ave$, and the two photon numbers are independent. Hence the probability a frame being selected is lower-bounded by
\begin{equation}\label{eq:probselect}
	\left(\eta  b \ave  \e^{-\eta  b \ave}\right)\cdot \left(\e^{-b(1-\eta)\ave}\right) = \eta  b \ave - \eta  b^2 \ave^2 + o\left(\frac{1}{\log\nicefrac{1}{\ave}}\right).
\end{equation}
Multiplying \eqref{eq:probselect} with $H(\tilde{X}| \sigma^{\mathbb{E}^b})$ gives us the secret-key nats per frame, where we count both selected and unselected frames. Simple normalization then yields the photon efficiency
\begin{IEEEeqnarray}{rCl}
	r_\textnormal{S-2}(\eta,\ave)
	 & \ge & \frac{1}{{\eta  b \ave}} \cdot \left(\eta  b \ave - \eta  b^2 \ave^2 + o\left(\frac{1}{\log\nicefrac{1}{\ave}}\right)\right) \nonumber\\
	& & {} \cdot \left(\log\frac{1}{\ave} - \log\log\frac{1}{\ave} + o(1)\right)\\
	& = & \log\frac{1}{\ave} - \log\log\frac{1}{\ave} - 1 + o(1).
\end{IEEEeqnarray}

\subsection{Proof of Proposition~\ref{prp:S-3}}\label{app:S-3}

The second part of the secret key, which is generated in Step 5) in Scheme~\ref{Scheme_S-3}, is exactly the (whole) secret  key generated by Scheme~\ref{Scheme_S-2}, and hence contributes to the total photon efficiency by the right-hand side of~\eqref{eq:S-2}. It is clear that this is independent of the first part of the key, as the latter only contains information of the frame labels. We thus only need to evaluate the contribution to the photon efficiency from the first part of the key which is generated in Step~5).

Consider a block of $\ell$ length-$b$ frames. To compute the length of the first part of the key that can be obtained from these frames, we first consider the information leakage due to Bob's message to Alice. Note that $(\tilde{A}^\ell,\tilde{B}^\ell)$ is distributed i.i.d. in time, where
each pair $(A,B)$ has joint distribution according to a Z channel with 
\begin{subequations}\label{eq:PABtilde}
\begin{IEEEeqnarray}{rCl}
	\tilde q &\triangleq& P_A(1)  =  1 - \e^{-b \ave} \\
	\tilde \mu & \triangleq& P_{B|A}(1|1)  =  \frac{1-\e^{-\eta  b \ave}}{1-\e^{-b \ave}}.
\end{IEEEeqnarray}
\end{subequations}

The optimal Slepian-Wolf code for Bob to convey $\tilde{B}^{\ell}$ to Alice should contain, asymptotically, $H(\tilde{B}|\tilde{A})$ nats per frame \cite{slepianwolf73}. Let $M_B$ be the message which Bob sends to Alice, then
\begin{equation}
	H(M_B) = \ell  H(\tilde{B}|\tilde{A}) + \ell \epsilon
\end{equation}
where $\epsilon$ tends to zero as $\ell$ tends to infinity.

We next bound the information leakage due to the message which Alice sends to Bob. A simple upper bound is: for each frame where $\tilde{B}=1$, Alice needs to send Bob at most one bit. From \eqref{eq:PABtilde} we can obtain
\begin{equation}\label{eq:S-3_1}
	P_{\tilde{B}}(1) = 1-\e^{-\eta  b \ave}.
\end{equation}
Let $M_A$ be the message which Bob sends to Alice for $\ell$ frames, then
\begin{equation}\label{eq:S-3_2}
	H(M_A) \le \ell \left(1-\e^{-\eta  b \ave}\right) + \ell \epsilon.
\end{equation}

We finally consider Eve's quantum state from the optical channel. Denote this state over $\ell$ frames by $\rho^{\mathbb{E}^{b\ell}}$. Since, as shown in Section~\ref{sec:key_source_setup}, it is independent of Bob's measurement (direct detection) outcomes, and since $\tilde{B}$ is a function of Bob's measurement outcomes, we know that
\begin{equation}\label{eq:S-3_3}
	I\left(\tilde{B}^{\ell}; \rho^{\mathbb{E}^{b\ell}} \right) = 0.
\end{equation}

We now use \eqref{eq:S-3_1}, \eqref{eq:S-3_2} and \eqref{eq:S-3_3} to bound the length of the first part of the key for $\ell$ frames which, according to \eqref{eq:rennerkonig}, is given by
\begin{IEEEeqnarray}{rCl}
	\lefteqn{H(\tilde{B}^{\ell}| M_A,M_B, \rho^{\mathbb{E}^{b\ell}})}~~\nonumber\\
	 & = & H(\tilde{B}^{\ell}) - \underbrace{I(\tilde{B}^{\ell}; \rho^{\mathbb{E}^{b\ell}})}_{=0} - \underbrace{I(M_A,M_B; \tilde{B}^{b\ell}| \rho^{\mathbb{E}^{b\ell}})}_{\le H(M_A) + H(M_B)} \IEEEeqnarraynumspace\\
	& \ge & \underbrace{H(\tilde{B}^\ell)}_{=\ell  H(\tilde{B})} - \underbrace{H(M_A)}_{\le \ell(1-\e^{-\eta  b \ave}) + \ell \epsilon} - \underbrace{H(M_B)}_{=\ell  H(\tilde{B}|\tilde{A}) + \ell \epsilon}\\
	& \ge & \ell  H(\tilde{B}) - \ell \left(1-\e^{-\eta  b \ave}\right) - \ell  H(\tilde{B}|\tilde{A}) - 2  \ell \epsilon\\
	& = & \ell  I(\tilde{A};\tilde{B}) + \ell  \eta  b \ave - 2  \ell \epsilon + o(\ave).
\end{IEEEeqnarray}
Hence, for large enough $\ell$, the length of the first part of the key \emph{per frame} is given by
\begin{equation}\label{eq:S-3_4}
	I(\tilde{A};\tilde{B}) + \ell  \eta  b \ave + o(\ave).
\end{equation}
We next evaluate $I(\tilde{A};\tilde{B})$. Comparing the parameters \eqref{eq:PABtilde} to \eqref{eq:PAB}, we see that $I(\tilde{A};\tilde{B})$ is the same as $I(A;B)$ \eqref{eq:S-1_last}, replacing $\ave$ with $b \ave$, where, recalling \eqref{eq:blocklog}, 
\begin{equation}
	b \ave = \frac{1}{\log\nicefrac{1}{\ave}} + o(\ave).
\end{equation}
Thus,
\begin{IEEEeqnarray}{rCl}
	I(\tilde{A};\tilde{B}) & = & H_2(\e^{-\eta  b \ave}) - \left(1-\e^{-b \ave}\right) H_2\left(\frac{1-\e^{-\eta  b \ave}}{1-\e^{-b \ave}}\right) \IEEEeqnarraynumspace\\
	& = & \eta  b \ave \log \log \frac{1}{\ave} + \eta  b \ave-b \ave H_2(\eta)\nonumber\\
	& & {} +o\left(\frac{1}{\log\nicefrac{1}{\ave}}\right).\label{eq:S-3_5}
\end{IEEEeqnarray}
We can now compute the photon efficiency coming from the first part of the secret key in Scheme~\ref{Scheme_S-3} by dividing \eqref{eq:S-3_4} by $\eta  b \ave$ (the average number of photons Bob detects per frame), and by using \eqref{eq:S-3_5}. This photon efficiency is at least
\begin{equation} \label{eq:S-3_6}
	\log\log\frac{1}{\ave} + 1 - \frac{H_2(\eta)}{\eta} + o(1).
\end{equation}
Adding \eqref{eq:S-3_6} to the right-hand side of \eqref{eq:S-2}, i.e., to the photon efficiency coming from the second part of the secret key, we conclude that
\begin{equation}
	r_\textnormal{S-3} (\eta,\ave) \ge \log\frac{1}{\ave}  - \frac{H_2(\eta)}{\eta} + o(1).
\end{equation}

\section*{Acknowledgments}
The authors thank Nivedita Chandrasekaran and Jeffrey Shapiro for helpful discussions, and the anonimous reviewers for useful comments.

\bibliographystyle{IEEEtran}           
\bibliography{/Volumes/Data/wang/Library/texmf/tex/bibtex/header_short,/Volumes/Data/wang/Library/texmf/tex/bibtex/bibliofile}

\end{document}